\documentclass[twocolumn]{IEEEtran}
\pdfoutput=1
\usepackage{amsthm,amssymb,graphicx,multirow,amsmath,color,amsfonts}%,ulem}
\usepackage[update,prepend]{epstopdf}
\usepackage[noadjust]{cite}
\usepackage[latin1]{inputenc}
\usepackage{tikz}
\usetikzlibrary{arrows,calc}		% Optional ticks libraries
\usepackage{bbm} % for \mathbbm{1}
\usepackage{pdfpages}
\usepackage{tabulary}
\usepackage{multirow}

\catcode`,\active

\catcode`\,12

\def\nbs{{\mathbf{s}}}

\def\nb0{{\mathbf{0}}}
\def\nb1{{\mathbf{1}}}

\def\ncalB{{\mathcal{B}}}

\def\ncalF{{\mathcal{F}}}

\def\ncalK{{\mathcal{K}}}

\def\ncalN{{\mathcal{N}}}

\def\ncalP{{\mathcal{P}}}

\def\ncalR{{\mathcal{R}}}

\def\ncalX{{\mathcal{X}}}

\def\nbbE{{\mathbb{E}}}

\def\nbbP{{\mathbb{P}}}

\def\nbbR{{\mathbb{R}}}

\newtheorem{lemma}{Lemma}

\newtheorem{definition}{Definition}

\newtheorem{theorem}{Theorem}

\newtheorem{remark}{Remark}

\def\PPPok{\Phi_k}							% Original PPP
\def\PPPo{\Phi}
\def\PPPek{\Phi_k^{\rm (e)}	}	% Equivalent PPP
\def\PPPe{\Phi^{\rm (e)}}
\def\PPPoj{\Phi_j}							% Original PPP
\def\PPPej{\Phi_j^{\rm (e)}	}	% Equivalent PPP

\def\densityok{\lambda_k}

\def\densityek{\lambda_k^{\rm (e)}}
\def\intensityok{\Lambda_k}
\def\intensityek{\Lambda_k^{\rm (e)}}

\def\rc{\mathtt{R_c}}   % rate coverage

\def\T{\beta}							% Threshold = \beta_i
\def\sinr{\mathtt{SINR}}			% Signal to interference plus noise ratio

\def\sir{\mathtt{SIR}}

\allowdisplaybreaks % Allows breaking of eqnarray over multiple pages (avoids unnecessary blanks in the document before eqnarray)
\begin{document}

%%%% Code to Include Cover Letter %%%%%
%\includepdf{Cover_Letter}
%\setcounter{page}{1}
%%%%%%%%%%%%%%%%%%%%%%%%%%%%%%%%%%%%%%%

\graphicspath{{./Figures/}}
\title{Downlink Rate Distribution in Heterogeneous Cellular Networks under Generalized Cell Selection
%Downlink Rate Distribution in HetNets under Generalized Cell Selection
}
%\author{Harpreet S. Dhillon et al. \thanks{
\author{Harpreet S. Dhillon and Jeffrey G. Andrews\thanks{
H. S. Dhillon and J. G. Andrews are with the Wireless Networking and Communications Group (WNCG), The University of Texas at Austin, TX, USA (email: dhillon@utexas.edu and jandrews@ece.utexas.edu). \hfill Manuscript updated: \today.
} }

\maketitle

\begin{abstract}
Considering both small-scale fading and long-term shadowing, we characterize the downlink rate distribution at a typical user equipment (UE) in a heterogeneous cellular network (HetNet), where shadowing, following any general distribution, impacts cell selection while fading does not. Prior work either ignores the impact of channel randomness on cell selection or lumps all the sources of randomness into a single variable, with cell selection based on the instantaneous signal strength, which is unrealistic. As an application of the results, we study the impact of shadowing on load balancing in terms of the optimal per-tier selection bias needed for rate maximization.
\end{abstract}

\section{Introduction}
Opportunistic capacity-driven deployment of small cells is recognized as a key solution to keep up with the increasing capacity demand from cellular networks~\cite{AndJ2013}. This paradigm shift in deployments has impacted cellular networks in many ways, two of the prominent ones being: i) the existing macrocells are now joined by femtocells, picocells and distributed antennas, thus increasing the disparity in the BS capabilities, such as transmit power and backhaul capacity, 
and ii) the unplanned deployment of small cells has increased the uncertainty in the BS locations. In the pursuit of accurately capturing these new deployment trends for the analysis of HetNets, random spatial models have emerged as an attractive option~\cite{DhiGanJ2012}. Assuming Poisson Point Process (PPP) model for the BS locations further lends tractability and tools from stochastic geometry have been used to study various aspects of HetNets, see~\cite{ElSHosJ2013} for a survey.

Despite the success of these models, there remain several shortcomings that need to be addressed for realistic performance assessment. One of them is the simplistic set of assumptions for channel and cell selection models. With a key exception of~\cite{KeeBlaC2013}, prior work either ignores the impact of shadowing on cell selection and assumes that a UE always connects to one of the closet BSs of each tier~\cite{JoSanJ2012,DiGuiJ2013}, or lumps all the channel randomness into a single random variable and assumes that cell selection is based on the maximum instantaneous received power~\cite{DhiGanJ2012,MadResJ2012a}. Due to these simplifications, neither of these models is able to capture the fact that the long-term effects such as distance-based path loss and shadowing impact cell selection, while small-scale fading does not. 

Although this letter is in the same spirit as~\cite{KeeBlaC2013}, the main focus of~\cite{KeeBlaC2013} is on the downlink signal-to-interference-plus-noise ratio ($\sinr$), whereas we focus on the downlink data rate that additionally depends upon the load on each BS class. For instance, to maximize downlink rate, it may be preferable to connect to small cells even when they offer poor $\sinr$ in some cases, because owing to their smaller load they can more than compensate by offering a large percentage of time-frequency resources to each UE. 
%was on coverage probability which is solely a function of signal-to-interference-plus-noise ratio ($\sinr$).  Although the $\sinr$ partially determines the data rate, a more important factor, particularly in a HetNet, is the load on a BS. Small cells may offer poor $\sinr$ in some cases, but can more than compensate by offering a large percentage of its time-frequency resources to each UE due to a smaller load~\cite{SinDhiJ2013}. Thus, the downlink rate distribution, as formulated in this letter, depends both on the $\sinr$ and the load on each BS class. 
Leveraging the same general idea of {\em propagation (process) invariance}, as discussed in~\cite{KeeBlaC2013,BlaKeeJ2013}, we show that in addition to the $\sinr$ distribution, the service area approximations resulting from multiplicatively weighted Poisson Voronoi tessellation, and hence the load on each BS class~\cite{SinDhiJ2013}, can be easily extended to the general cell selection model introduced in this letter. Our analysis concretely demonstrates that the effect of shadowing on downlink rate and related metrics, such as rate optimal cell-selection bias, can be equivalently captured by appropriately scaling the transmit powers of each BS class and then simply selecting one of the BSs that are closest in each tier for service.%This insight immediately generalizes prior works on load balancing, such as~\cite{SinDhiJ2013}.

\section{System Model} \label{sec:sysmod}
Consider a $K$-tier HetNet with $K$ classes of BSs, differing in terms of the transmit power $P_k$, deployment density $\densityok$, and cell-selection bias $B_k$. For notational simplicity, define $\ncalK = \{1,2,\ldots,K\}$. The locations of the $k^{th}$ tier BSs are modeled by an independent PPP $\PPPok$ of density $\lambda_k$. Define $\PPPo = \cup_{k \in \ncalK} \PPPok$. For resource allocation, consider orthogonal partitioning of resources, e.g., time-frequency resource blocks in orthogonal frequency division multiple access (OFDMA), where each resource block is allocated to one UE, and hence there is no intra-cell interference. Modeling UE locations by an independent PPP $\Phi_u$ of density $\lambda_u$, the downlink analysis is performed at a typical UE assumed to be located at the origin~\cite{BacBlaB2009}. %The user locations are also modeled by an independent PPP $\Phi_u$ of density $\lambda_u$. This analysis is facilitated by Slivnyak's theorem, which states that the properties observed at a typical point of $\Phi_u$ are the same as those observed by the point at the origin in point process $\Phi_u \cup \{0\}$~\cite{StoKenB1995}. 
The received power at a typical UE from a $k^{th}$ tier BS located at $x_k \in \PPPok$ in a given resource block is
\begin{align}
P(x_k) = P_k h_{kx_k} \ncalX_{kx_k} \|x_k\|^{-\alpha},
\label{eq:Pow_received}
\end{align}
where $h_{kx_k} \sim \exp(1)$ models Rayleigh fading, $\ncalX_{kx_k}$ models shadowing, and $\|x_k\|^{-\alpha}$ represents standard power-law path loss with exponent $\alpha$. Note that since $h_{kx_k}$ and $\ncalX_{kx_k}$ are both independent of the location of the BS, we will drop $x_k$ from the subscript and denote the two random variables by $h_k$ and $\ncalX_k$, whenever the location of the BS is clear form the context. In the same spirit as~\cite{KeeBlaC2013,MadResJ2012a}, our analysis is capable of handling any general distribution for $\ncalX_{k}$ as long as $\nbbE\left[\ncalX_{k}^{\frac{2}{\alpha}} \right] < \infty$. The origins of this restriction will be discussed later in this section. The most common assumption for large scale shadowing distribution is lognormal, where $\ncalX_k = 10^{\frac{X_k}{10}}$ such that $X_{k} \sim \ncalN(\mu_k, \sigma_k^2)$, where $\mu_k$ and $\sigma_k$ are respectively the mean and standard deviation in dB of the shadowing channel power. Using the moment generating function (MGF) of Gaussian distribution, the fractional moment is $\nbbE\left[\ncalX_k^{\frac{2}{\alpha}} \right] = \exp \left(\frac{\ln 10}{5} \frac{m_k}{\alpha} + \frac{1}{2} \left(\frac{\ln 10}{5} \frac{\sigma_k}{\alpha} \right)^2 \right)$, which is clearly finite if both the mean and standard deviation of the normal random variable $X_{k}$ are finite.

Since fading gain $h_{x}$ changes over much smaller time-scale, and in a frequency selective channel (such as one using OFDM) can be averaged or mitigated in the frequency domain, we assume that it does not impact cell selection. Each UE connects to the BS that provides the highest long-term {\em biased received power}, as explained below. Denote the location of the candidate $k^{th}$ tier serving BS by $x^*_k$, i.e., $x^*_k = \arg \max_{x \in \PPPok} P_k \ncalX_x \|x\|^{-\alpha}$. From these $K$ candidate serving BSs, a typical UE connects to $x^* = \arg \max_{x\in \{x^*_k\} } B_k P_k \ncalX_x \|x\|^{-\alpha}$,
where $B_k > 0$ is the selection bias introduced to expand the range of small cells to balance load across the network~\cite{SinDhiJ2013}. The inclusion of shadowing in cell selection is facilitated by displacement theorem~\cite{BacBlaB2009}, where the key insight is to express the received power given by~\eqref{eq:Pow_received} as
%\begin{align}
$P(x_k) = P_k h_{kx_k} \|\ncalX_{kx_k}^{-\frac{1}{\alpha}} x_k\|^{-\alpha}$,
%\end{align}
where the long-term shadowing effects can be interpreted as a random displacement of the location of the BS originally placed at $x_k \in \PPPok$. We make this notion precise in the following Lemma. Also see~\cite{KeeBlaC2013,MadResJ2012a} for the application of this general idea to handle general shadowing or fading distributions in slightly different setups. 

\begin{lemma} \label{lem:mapping}
For a homogeneous PPP $\PPPok \subset \nbbR^2$ with density $\densityok$, if each point $x \in \PPPok$ is transformed to $y \in \nbbR^2$ such that $y = \ncalX_{k}^{-\frac{1}{\alpha}} x$, where $\{\ncalX_k\}$ are i.i.d., such that $\nbbE\left[\ncalX_{k}^{\frac{2}{\alpha}} \right] < \infty$, the new point process $\PPPek \subset \nbbR^2$ defined by the transformed points $y$ is also a homogeneous PPP with density $\densityek = \densityok \nbbE\left[\ncalX_{k}^{\frac{2}{\alpha}} \right]$.
\end{lemma}

\begin{proof}
Let $\ncalB(\nbbR^2)$ be the Borel $\sigma$-algebra on $\nbbR^2$. For $A \in \ncalB(\nbbR^2)$, the intensity measure $\Lambda(A)$ of a homogeneous PPP $\PPPok$ is $\intensityok(A) = \densityok |A|$, where $|A|$ denotes the Lebesgue measure of A. By displacement theorem~\cite[Theorem 1.3.9]{BacBlaB2009}, the transformation of a PPP $\PPPok$ with probability kernel $p(x,A)$ is a PPP with intensity measure:
\begin{align}
\intensityek(A) &= \int_{\nbbR^2} p(x,A) \intensityok({\rm d}x) 
\stackrel{(a)}{=} \nbbE  \int_{\nbbR^2} \nb1(\ncalX_k^{-\frac{1}{\alpha}} x \in A) \densityok {\rm d}x \nonumber \\
&= \nbbE \left[ \int_{\nbbR^2} \nb1\left(x \in A \ncalX_k^{\frac{1}{\alpha}}\right) \densityok {\rm d}x \right]
= \densityok |A| \nbbE \left[\ncalX_k^{\frac{2}{\alpha}}\right], \nonumber
\end{align}
where $(a)$ follows by using the kernel specific to this Lemma.  
Since $\{\ncalX_k\}$ are i.i.d. and independent of the location $x$, setting $|A| = {\rm d} y$, we get $\intensityek(A)({\rm d} y) = \densityok \nbbE \left[\ncalX_k^{\frac{2}{\alpha}}\right] {\rm d} y = \densityek {\rm d}y$. For a PPP, we need its intensity measure to be locally finite, which leads to the condition $\densityek = \densityok \nbbE \left[\ncalX_k^{\frac{2}{\alpha}}\right] < \infty$.
\end{proof}

An immediate consequence of this Lemma is the characterization of received power in terms of the equivalent PPP $\PPPek$ with density $\densityek = \densityok \nbbE \ncalX_{k}^{\frac{2}{\alpha}} $. Defining $y_k = \ncalX_{kx_k}^{-\frac{1}{\alpha}} x_k$, the received power can be equivalently expressed as 
$P(y_k) = P_k h_{k y_k} \|y_k\|^{-\alpha}$, using which the location of the candidate serving BS in $k^{th}$ tier can be equivalently expressed as
$y_k^* = \arg \max_{y \in \PPPek} P_k \|y\|^{-\alpha}$. Note that $y_k^*$ is simply the closest point to the origin of the equivalent point process $\PPPek$. The location of the serving BS in the equivalent PPP $\PPPe = \cup_{k \in \ncalK} \PPPek$ can be similarly expressed as 
%\begin{align}
$y^* = \arg \max_{y \in \{y_k^*\} } B_k P_k \|y\|^{-\alpha}$.
%\end{align}
For notational simplicity define $\nbs \in \nbbR^K$, such that $\nbs(k) \in \{0,1\}$, $\sum_{k \in \ncalK} \nbs(k) = 1$, and
%\begin{align}
$\nbs(k) = \nb1(x^* = x_k^*) = \nb1(x^* \in \PPPok)$, 
%\end{align}
which implies that the $k^{th}$ element of $\nbs$ takes value $1$ if the serving BS belongs to $k^{th}$ tier. We ignore thermal noise, i.e., network is interference-limited, and assume a full-buffer model for the interfering BSs~\cite{DhiGanJ2012}, i.e., all the interferers are always active. The signal-to-interference ratio $(\sir)$ at the typical UE when $\nbs(k)=1$ is
\begin{align}
\sir(x^*) &= \frac{P_k h_{kx^*} \ncalX_{kx^*} \|x^*\|^{-\alpha}}{\sum_{j \in \ncalK} \sum_{z \in \PPPoj\setminus \{x^*\}} P_j h_{jz} \ncalX_{jz} \|z\|^{-\alpha}} \\
&\stackrel{d}{=} \frac{P_k h_{ky^*}  \|y^*\|^{-\alpha}}{\sum_{j \in \ncalK} \sum_{z \in \PPPej\setminus \{y^*\}} P_j h_{jz} \|z\|^{-\alpha}} = \sir(y^*), \nonumber
\end{align}
where $d$ denotes equivalence in distribution, which follows from Lemma~\ref{lem:mapping}. Due to this equivalence, the results based solely on $\sir$ or $\sinr$ distributions, such as coverage probability, derived under the assumption that a typical UE always connects to one of the BSs that are closest in each tier, e.g.,~\cite{JoSanJ2012}, can be easily extended to the general selection model by considering equivalent BS densities $\{\densityek\}$. This has also been independently shown for coverage probability in~\cite{KeeBlaC2013}. In the next section, we establish a similar equivalence for downlink rate distribution, that additionally depends upon the BS load.

\section{Downlink Rate Distribution}
In this section, we generalize the main premise of~\cite{SinDhiJ2013}, and characterize the downlink rate coverage under generalized cell-selection model introduced in the previous section.

\begin{definition}[Rate coverage]
Rate coverage $\rc$ is the probability that the downlink rate $\ncalR$ achievable at a typical UE is higher than a predefined lowest rate $T$ required by a given application, i.e., $\rc = \nbbP(\ncalR > T)$. Being the complementary cumulative distribution function (CCDF), $\rc$ completely characterizes the rate distribution. %It can also be interpreted in terms of the number of users or the average area where the minimum rate demand, characterized by threshold $T$, is met.
\end{definition}

We term the serving BS $x^* \in \PPPo$ of a typical UE as a ``tagged'' BS. Denote the number of UEs served by the tagged BS by $\Psi_k$, where subscript $k$ is for the tier to which this BS belongs. Clearly, $\Psi_k$ is a random variable with the following two sources of randomness: i) the area of the region that the tagged BS serves, in short service area, and ii) conditioned on the area of the service region, the number of UEs served by the tagged BS is a Poisson random variable. For tractability, we assume that each BS allocates equal time-frequency resources to its UEs, i.e., each UE gets rate proportional to the spectral efficiency of its downlink channel from the serving BS. For total effective bandwidth $W$ Hz, the downlink rate in bits/sec of a typical UE when it connects to a $k^{th}$ tier BS is
\begin{align}
\ncalR_k = \frac{W}{\Psi_k} \log_2 \left(1 + \sir(x_k^*) \right).
\end{align}
%which in addition to the link quality, characterized by $\sir(x_k^*)$, also depends upon the number of UEs $\Psi_k$ served by the tagged BS. 
Note that $\sir(x_k^*)$ and $\Psi_k$ are in general correlated, e.g., when $\|x_k^*\|$ is large, the serving BS is far from the typical UE. This information skews the distribution of the service area of the tagged BS, and hence of $\Psi_k$,  towards larger values. However, characterizing the joint distribution of $\Psi_k$ and $\sir(x_k^*)$ is out of the scope of this paper. For tractability, we assume the two random variables to be independent, which does not compromise the accuracy of our analysis~\cite{SinDhiJ2013}. 
%\begin{assumption}[Independence] \label{as:independence}
%We assume that the random variables $\sir(x_k^*)$ and $\Psi_k$ are independent.
%\end{assumption} 
Under this assumption, the rate coverage $\rc$ is
\begin{align}
&\rc = \nbbP[\ncalR > T] \stackrel{(a)}{=} \sum_{k \in \ncalK} \nbbP[\ncalR > T| \nbs(k)=1] \nbbP[\nbs(k)=1]\\
%&\stackrel{(b)}{=} \sum_{k=1}^K \nbbE_{\Psi_k} \nbbP \left( \frac{1}{\Psi_k} \log_2 (1 + \sir (x_k^*)) > T \right) \ncalP_k \\
&= \sum_{k \in \ncalK} \nbbE_{\Psi_k} \nbbP \left( \ncalR_k > T \right) \nbbP[\nbs(k)=1] \\
&= \sum_{k \in \ncalK} \nbbE_{\Psi_k} \nbbP \left(\sir(x_k^*) > 2^{\frac{T}{W} \Psi_k} - 1 \right) \nbbP[\nbs(k)=1] \\
%&\stackrel{(c)}{=} \sum_{k \in \ncalK} \nbbE_{\Psi_k} \frac{1}{1+ \ncalF (2^{T \Psi_k} - 1 , \alpha, 1) } \ncalP_k \\
&\stackrel{(b)}{=}  \sum_{k\in \ncalK} \sum_{n=1}^{\infty} \underbrace{\nbbP \left(\sir(x_k^*) > \T_n \right)}_\textrm{Conditional $\sir$ distribution} 
\underbrace{\nbbP(\Psi_k = n)}_\textrm{Load}
\underbrace{\nbbP[\nbs(k)=1]}_\textrm{Selection probability}, \label{eq:Rc_defn}
\end{align} 
where $(a)$ follows from the total probability theorem, and $(b)$ follows by defining $\T_n = 2^{\frac{T}{W} n} - 1$ for notational simplicity. We now compute the three probability terms starting with the selection probability, which we denote by $\ncalP_k$.
\begin{lemma}[Selection probability]\label{lem:selection}
The probability that a typical UE connects to a $k^{th}$ tier BS is given by
\begin{align}
\ncalP_k = \nbbP(\nbs(k)=1) = \frac{ \lambda_k \nbbE \left[\ncalX_k^{\frac{2}{\alpha}} \right] B_k^{\frac{2}{\alpha}} P_k^{\frac{2}{\alpha}}}{\sum_{j \in \ncalK}  \lambda_j \nbbE \left[\ncalX_j^{\frac{2}{\alpha}} \right] B_j^{\frac{2}{\alpha}} P_j^{\frac{2}{\alpha}}}.
\end{align}
\end{lemma}
\begin{IEEEproof}
The selection probability is
\begin{align}
\ncalP_k =\nbbP(\nbs(k)=1) \stackrel{(a)}{=} \nbbP (x^* = x_k^*) \stackrel{(b)}{=} \nbbP (y^* = y_k^*),
%\nbbP(\nbs(k)=1) &\stackrel{(a)}{=} \nbbP(P(x_k^*) > \max_{\substack{j \in \ncalK \\ j \neq k}} P(x_j^*))\\
%&\stackrel{(b)}{=} \nbbP(P(y_k^*) > \max_{\substack{j \in \ncalK \\ j \neq k}} P(y_j^*)), 
\end{align}
where $\{x_k^*\}$ in $(a)$ is the set of candidate serving BSs in $\PPPo$, $\{y_k^*\}$ in $(b)$ is the set of candidate serving BSs in $\PPPe$, and $(b)$ additionally follows from Lemma~\ref{lem:mapping}. Recall that the candidate serving BS $y_k^*$ is the closest point of the PPP $\PPPek$ to the origin, which reduces to the same setup as~\cite{JoSanJ2012}. The rest of the proof follows from the Lemma 1 of~\cite{JoSanJ2012} using the fact that the density of $\PPPek$ is $\densityek = \lambda_k \nbbE \left[\ncalX_k^{\frac{2}{\alpha}} \right]$. 
\end{IEEEproof}

We now derive the conditional $\sir$ distribution, i.e., $\sir$ distribution conditioned on $\nbs(k)=1$. The proof follows directly from Theorem 1 of~\cite{JoSanJ2012} after invoking displacement theorem as was done for Lemma~\ref{lem:selection}, and is hence skipped.

\begin{lemma}[Conditional $\sir$ distribution] \label{lem:pertierPc}
The conditional $\sir$ distribution is $\nbbP(\sir(x_k^*) > \T)=$
\begin{align}
%\nbbP(\sir(x^*) > \T | \nbs(k)=1) =  \frac{1}{1+\ncalF(\T, \alpha, 1)},
\frac{\sum_{j\in\ncalK} \lambda_j \nbbE \left[\ncalX_j^{\frac{2}{\alpha}} \right] B_j^{\frac{2}{\alpha}} P_j^{\frac{2}{\alpha}}}
{\sum_{j\in\ncalK} \lambda_j \nbbE \left[\ncalX_j^{\frac{2}{\alpha}} \right] P_j^{\frac{2}{\alpha}} \left[B_j^{\frac{2}{\alpha}} + B_k^{\frac{2}{\alpha}} \ncalF\left(\T, \alpha, \frac{B_j}{B_k} \right)\right]  }
\end{align}
%\begin{align}
%\frac{2 \pi \rho_k \lambda_k \nbbE[\ncalX_k^{\frac{2}{\alpha}}]}{\ncalP_k} \int\limits_{0}^{\infty} y \exp \left(-\frac{\T W}{P_k} y^{\alpha}  - \pi \sum_{j=1}^K \ncalC_j y^2 \right) {\rm d} y,
%\label{eq:pertierPc}
%\end{align}
where
\begin{align}
%\ncalC_j  &= \rho_j \lambda_j \nbbE[\ncalX_j^{\frac{2}{\alpha}}] \left( \frac{P_j}{P_k} \right)^{\frac{2}{\alpha}} (1 + \ncalF(\T, \alpha, 1))\\
\ncalF(\T, \alpha, z) &= \left( \frac{2 \T z^{\frac{2}{\alpha}-1}}{\alpha - 2} \right) {}_2F_1 \left[1,1-\frac{2}{\alpha}, 2 - \frac{2}{\alpha}, -\frac{\T}{z} \right], \nonumber
\end{align}
and ${}_2F_1[a,b,c,z] = \frac{\Gamma(c)}{\Gamma(b) \Gamma(c-b)} \int_{0}^1 \frac{t^{b-1} (1-t)^{c-b-1}}{(1-tz)^a} {\rm d} t$ denotes the Gauss hypergeometric function. 
\end{lemma}

For the distribution of load $\Psi_k$, we use the approximation proposed in Lemma 3 of~\cite{SinDhiJ2013}. The main idea is to approximate the service area of a typical $k^{th}$ tier BS by the area of a typical Poisson Voronoi with the same mean $\frac{\ncalP_k}{\lambda_k}$. Now since a typical UE has a higher chance of selecting a BS with bigger service area, the area of the tagged BS is biased towards being larger than a typical BS of the same tier. This is similar to the waiting bus paradox associated with Point processes in $\nbbR$. Accounting for this bias, the load distribution is given below. The proof is the same as Lemma 3 of~\cite{SinDhiJ2013} with the understanding that the effect of shadowing on cell selection is captured by $\ncalP_k$.

\begin{lemma}[Load on tagged BS]    \label{lem:load_tagged}
The distribution of the load served by $x_k^*$ is $\nbbP(\Psi_k = n+1) = $
\begin{align}
\frac{3.5^{3.5}}{n!} \frac{\Gamma(n+4.5)}{\Gamma(3.5)} \left(\frac{\lambda_u \ncalP_k}{\lambda_k} \right)^n \left(3.5 +  \frac{\lambda_u \ncalP_k}{\lambda_k}  \right)^{-(n+4.5)}. \nonumber
\end{align}
The mean load is $\nbbE[\Psi_k] = 1 + 1.28 \frac{\lambda_u \ncalP_k}{\lambda_k}$.
%\begin{align}
%\nbbE[\Psi_k] = 1 + 1.28 \frac{\lambda_u \ncalP_k}{\lambda_k} = 1 + \frac{1.28 \lambda_u \nbbE \left[\ncalX_k^{\frac{2}{\alpha}} \right] B_k^{\frac{2}{\alpha}} P_k^{\frac{2}{\alpha}}}{\sum_{j =1}^K  \lambda_j \nbbE \left[\ncalX_j^{\frac{2}{\alpha}} \right] B_j^{\frac{2}{\alpha}} P_j^{\frac{2}{\alpha}}}. \nonumber\label{eq:meanload_tagged}
%\end{align}
\end{lemma}

Substituting Lemmas~\ref{lem:selection},~\ref{lem:pertierPc}, and~\ref{lem:load_tagged} in \eqref{eq:Rc_defn}, we get a fairly simple expression for rate coverage.

\begin{theorem}[Rate coverage] \label{thm:Rc}
The rate coverage is $\rc =$
\begin{align}
 \sum_{k=1}^K \sum_{n\geq 0} 
 \frac{\lambda_k \nbbE \left[\ncalX_k^{\frac{2}{\alpha}} \right] B_k^{\frac{2}{\alpha}} P_k^{\frac{2}{\alpha}} \nbbP(\Psi_k = n+1)}
{\sum_{j\in\ncalK} \lambda_j \nbbE \left[\ncalX_j^{\frac{2}{\alpha}} \right] P_j^{\frac{2}{\alpha}} \left[B_j^{\frac{2}{\alpha}} + B_k^{\frac{2}{\alpha}} \ncalF\left(\T_{n+1}, \alpha, \frac{B_j}{B_k} \right)\right]  }  \nonumber   %\\
%\times \frac{3.5^{3.5}}{n!} \frac{\Gamma(n+4.5)}{\Gamma(3.5)} \left(\frac{\lambda_u \ncalP_k}{\lambda_k} \right)^n \left(3.5 +  \frac{\lambda_u \ncalP_k}{\lambda_k}  \right)^{-(n+4.5)}, \nonumber
\end{align}
where $\ncalP_k$ is given by Lemma~\ref{lem:selection}, $\nbbP(\Psi_k = n+1)$ by Lemma~\ref{lem:load_tagged}, and recall that $\T_{n+1} = 2^{\frac{T}{W} (n+1)} -1$.
\end{theorem}

We will validate the load approximation and study the effect of shadowing on load balancing in the next section. This section is concluded with the following remarks.
\begin{remark}[Invariance] \label{rem:invariance}
If the shadowing distribution is such that $\nbbE \left[\ncalX_k^{\frac{2}{\alpha}} \right] = c$, for all $k \in \ncalK$, the downlink rate distribution is invariant to the shadowing distributions of all the tiers.
\end{remark}

\begin{remark}[Equivalent HetNet model] \label{rem:equivalent}
A HetNet model with $k^{th}$ tier transmit power $\left( \nbbE \left[\ncalX_k^{\frac{2}{\alpha}} \right] \right)^\frac{\alpha}{2} P_k$, no shadowing, and cell selection based on average biased receive power with selection bias $\{B_k\}$, leads to the same expression for rate coverage as given by Theorem~\ref{thm:Rc} for the generalized cell selection model. Due to this equivalence, the key results derived under no shadowing,  e.g., in~\cite{JoSanJ2012,SinDhiJ2013}, can be easily extended to the generalized cell selection model by appropriately scaling the transmit powers.
\end{remark}

\begin{figure}
\centering
\includegraphics[width=.95\columnwidth]{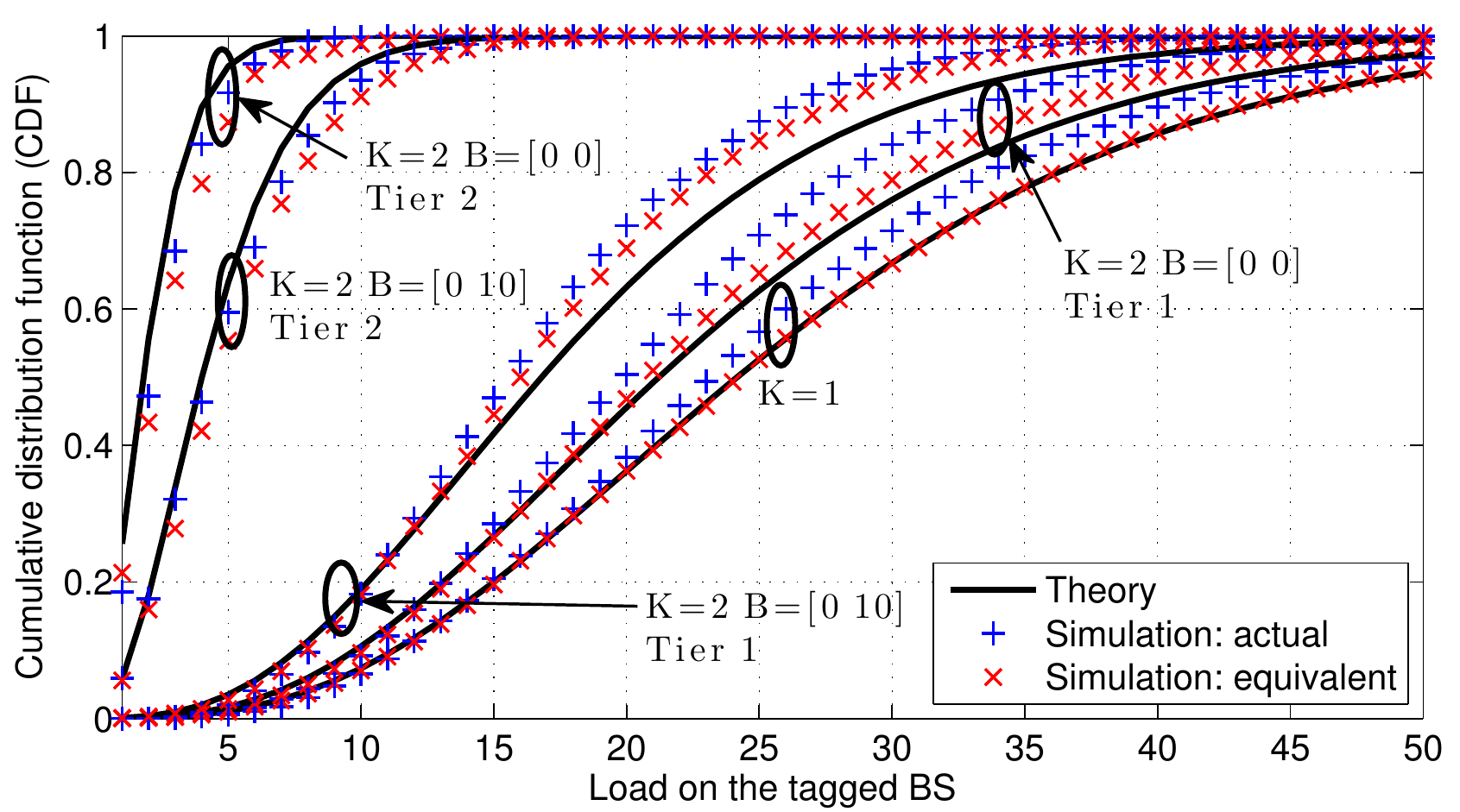} 
\includegraphics[width=.95\columnwidth]{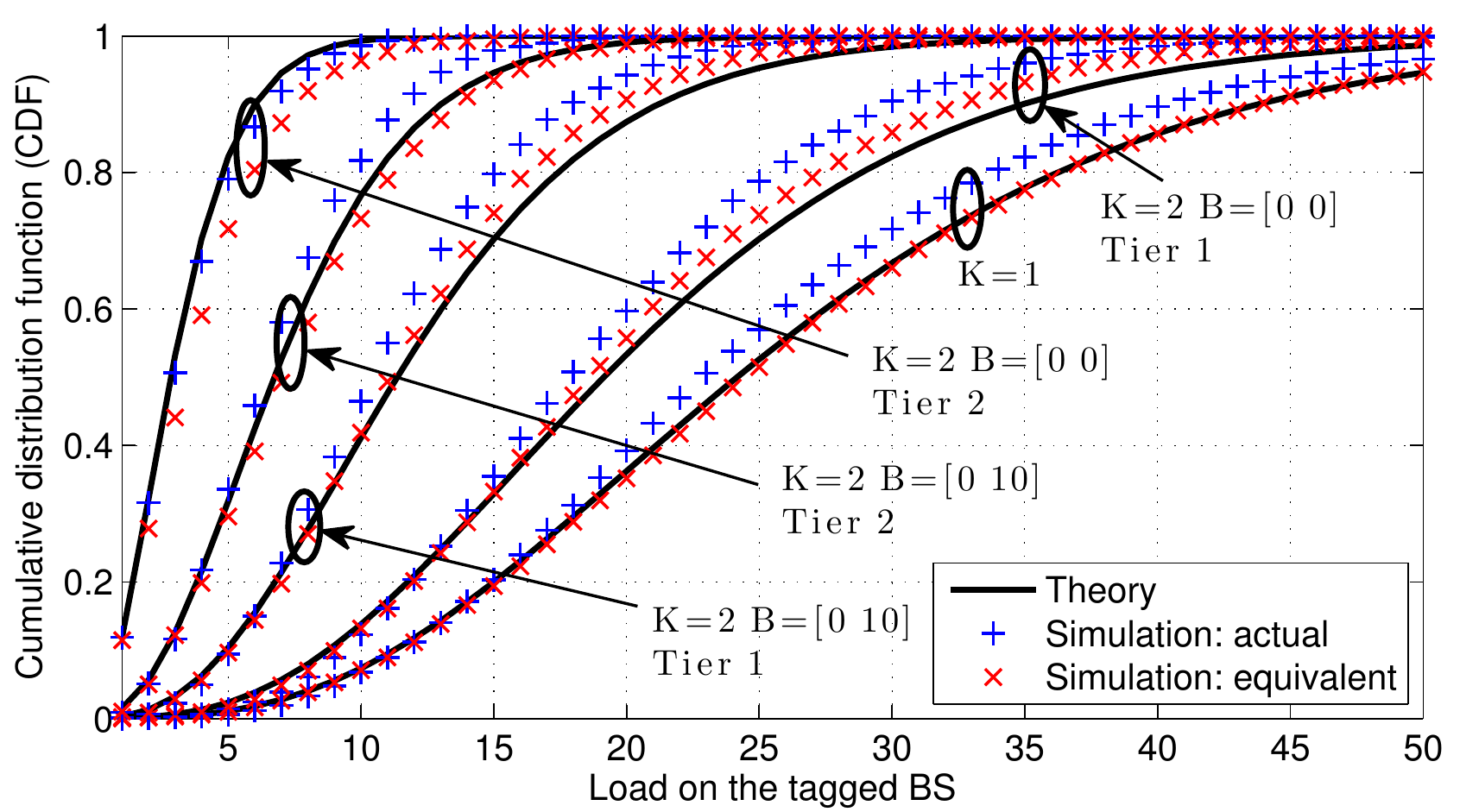}
\caption{CDF of load $\Psi_k$ with $\lambda_u = 20\lambda_1$, $\lambda_2 = 2\lambda_1$ for $K=2$, $\lambda_2=0$ for $K=1$, $\sigma = [4\ 4]$ dB {\em (first)} and $[4\ 8]$ dB {\em (second)}. $B$ is in dB.} 
\label{fig:load_tagged}
\end{figure}

\section{Numerical Results}
For numerical results, we consider a two tier HetNet, e.g., coexisting macro and pico cells, and assume that the shadowing distribution for each tier is log-normal with mean $\mu_k$ dB and standard deviation $\sigma_k$ dB. Throughout this section, we assume $\alpha=4$, $P_2 = P_1 - 23$ dB, $W = 10$ MHz, and $\mu = [0\ 0]$ dB. We first validate the load distribution given by Lemma~\ref{lem:load_tagged} in Fig.~\ref{fig:load_tagged}. In addition to the actual load under the generalized cell selection model, we also plot the load offered to the tagged BS under an equivalent model suggested in Remark~\ref{rem:equivalent}. We first note that the analytic approximation given by Lemma~\ref{lem:load_tagged} is fairly accurate, which along with the fact that the other components of rate expression, i.e., conditional $\sir$ distribution and selection probability, are exact, leads to a very tight approximation for rate distribution, as validated in Fig.~\ref{fig:Rate}. Comparing the rate distributions for two sub-figures of Fig.~\ref{fig:load_tagged}, we note that there is a natural balancing of load across tiers when $\nbbE \left[\ncalX_2^{\frac{2}{\alpha}} \right] > \nbbE \left[\ncalX_1^{\frac{2}{\alpha}} \right]$ compared to the baseline case of no shadowing, which by Remark~\ref{rem:equivalent} is equivalent to the case when $\nbbE \left[\ncalX_2^{\frac{2}{\alpha}} \right] = \nbbE \left[\ncalX_1^{\frac{2}{\alpha}} \right]$, as in the first sub-figure. This load balancing can be understood in terms of the equivalent model proposed in Remark~\ref{rem:equivalent}, i.e., in this case shadowing increases the effective transmit power of small cells relative to the baseline and hence expands their coverage areas. In Fig.~\ref{fig:Rate}, we plot the rate coverage and the fifth percentile rate, i.e., the rate value such that $95\%$ of the UEs achieve higher rate than this value. Both these results are consistent with the load balancing observations made in Fig.~\ref{fig:load_tagged}, e.g., the optimal selection bias that maximizes fifth percentile rate is smaller when $\nbbE \left[\ncalX_2^{\frac{2}{\alpha}} \right] > \nbbE \left[\ncalX_1^{\frac{2}{\alpha}} \right]$. Due to a smaller artificial bias, this case also achieves the highest rate. In Fig.~\ref{fig:Rate}, we also validate the rate expression by comparing it with the simulations and a special case in which the load on a tagged BS is assumed to be deterministic and equal to its mean $\nbbE[\Psi_k] $, given by Lemma~\ref{lem:load_tagged}.

\begin{figure}
\centering
\includegraphics[width=.95\columnwidth]{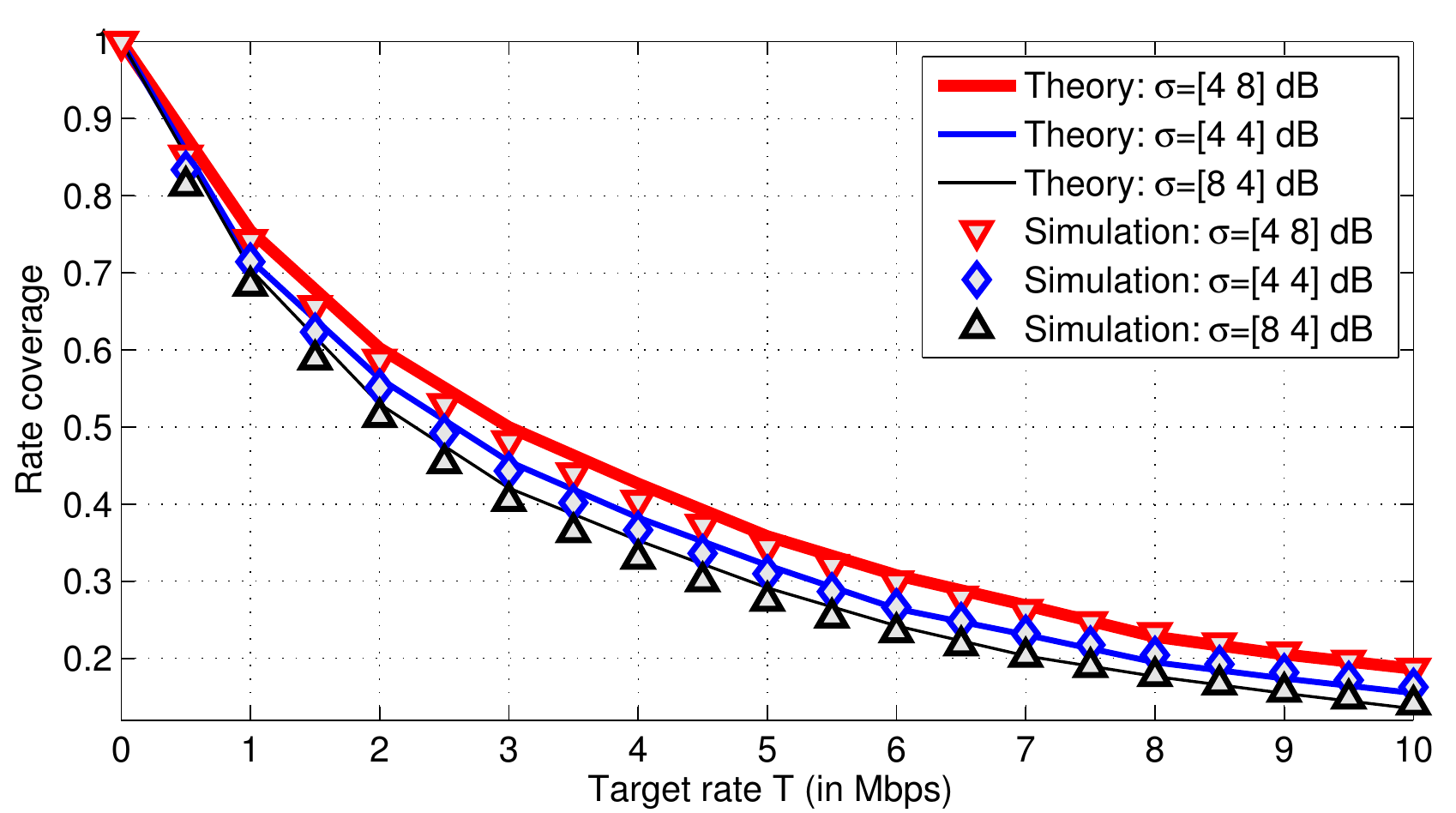}
\includegraphics[width=.95\columnwidth]{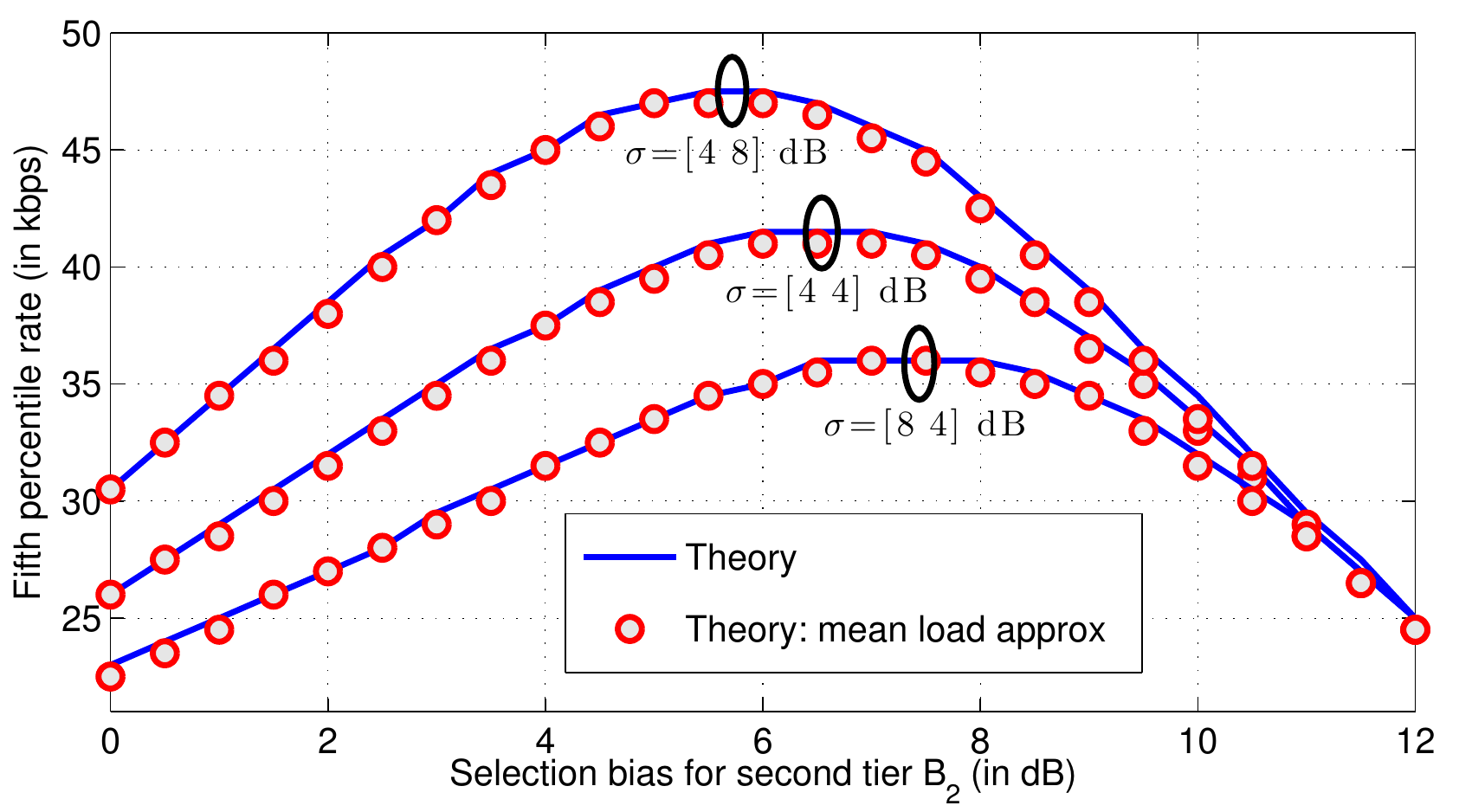}
\caption{{\em (first)} Rate coverage for $\lambda_2 = 5 \lambda_1$, $\lambda_u = 10\lambda_1$, $B=[0\ 5]$ dB. {\em (second)} Fifth percentile rate for $\lambda_2 = 5\lambda_1, \lambda_u = 40\lambda_1$, $B_1=0$ dB.}
\label{fig:Rate}
\end{figure}

\section{Conclusions}
We have derived the downlink rate distribution under a generalized cell-selection model, which explicitly differentiates between long-term channel effects such as shadowing and path-loss, and small-scale effects such as fading. We proposed an equivalent interpretation of this general cell selection model and showed that the effect of shadowing can be equivalently studied by appropriately scaling transmit powers. Using this equivalent interpretation, we studied the effect of shadowing on load balancing, and showed that in certain regimes shadowing naturally balances load across various tiers and hence reduces the need for artificial cell selection bias. 

\bibliographystyle{IEEEtran}
\bibliography{GeneralShadowing_Letter_v1.2}
\end{document}